\documentclass[runningheads]{llncs}
\usepackage{graphicx}
\usepackage{amsmath}

\usepackage[utf8x]{inputenc}
\usepackage{verbatim}
\usepackage{enumerate}
\usepackage{amssymb}
\usepackage{todonotes}
\usepackage{graphicx}
\usepackage{textcomp}
\usepackage{algorithm}
\usepackage{amsmath}

\newcommand{\DFS}{{\sf{DFS-Indexing}}}
\newcommand{\BICON}{{\sf{Bi-Connectivity-Indexing}}}
\newcommand{\EDGE}{{\sf{2-Edge-Connectivity-Indexing}}}
\newcommand{\STRONG}{{\sf{Strong-Connectivity-Indexing}}}
\newcommand{\SHORT}{{\sf{Shortest Path-Indexing}}}
\newcommand{\UCON}{{\sf{Undirected-Connectivity-Indexing}}}

\begin{document}
\title{Indexing Graph Search Trees and Applications \protect\footnote{This work was partially supported by JST CREST Grant Number JPMJCR1402, Japan.}}
\author{Sankardeep Chakraborty\inst{1}, Kunihiko Sadakane\inst{2}}

\authorrunning{Chakraborty and Sadakane}

\institute{RIKEN Center for Advanced Intelligence Project, Tokyo, Japan\\
\email {sankar.chakraborty@riken.jp}\\
\and 
The University of Tokyo, Tokyo, Japan,\\
\email {sada@mist.i.u-tokyo.ac.jp}\\
}

 \maketitle

\begin{abstract}
We consider the problem of compactly representing the Depth First Search (DFS) tree of a given undirected or directed graph having $n$ vertices and $m$ edges while supporting various DFS related queries efficiently in the RAM with logarithmic word size.
We study this problem in two well-known models: {\it indexing} and {\it encoding} models. While most of these queries can be supported easily in constant time using $O(n \lg n)$ bits\footnote{We use $\lg$ to denote logarithm to the base $2$.} of extra space, our goal here is, more specifically, to beat this trivial $O(n \lg n)$ bit space bound, yet not compromise too much on the running time of these queries. In the {\it indexing} model, the space bound of our solution involves the quantity $m$, hence, we obtain different bounds for sparse and dense graphs respectively. In the {\it encoding} model, we first give a space lower bound, followed by an almost optimal data structure with extremely fast query time. Central to our algorithm is a partitioning of the DFS tree into connected subtrees, and a compact way to store these connections. Finally, we also apply these techniques to compactly index the shortest path structure, biconnectivity structures among others.
\end{abstract}

\section{Introduction}
Depth First Search (DFS) is a very well-known method for visiting the vertices and edges of a directed or undirected graph. DFS differs from other ways of traversing the graph such as Breadth First Search (BFS) by the following DFS protocol: Whenever two or more vertices were discovered by the search method and have unexplored incident (out)edges, an (out)edge incident on the most recently discovered such vertex is explored first. This DFS traversal produces a rooted spanning tree (forest), called DFS tree (forest) along with assigning an index to every vertex $v$ i.e., the time vertex $v$ is discovered for the first time during DFS. We call it depth-first-index (DFI($v$)). Let $G=(V,E)$ be a graph on $n=|V|$ vertices and $m=|E|$ edges where $V =\{v_1, v_2, \cdots, v_n\}$. It takes $O(m+n)$ time to perform a DFS traversal of $G$ and to generate its DFS tree (forest) with DFIs of all the vertices. The DFS rule confers a number of structural properties on the resulting graph traversal that cause DFS to have a large number of applications. These properties are captured in the DFS tree (forest), and can be used crucially to design efficient algorithms for many basic and fundamental algorithmic graph problems, namely, biconnectivity~\cite{Tarjan72}, $2$-edge connectivity~\cite{Tarjan74}, strongly connected components~\cite{Tarjan72},  topological sorting~\cite{Tarjan72}, dominators~\cite{Tarjan741}, {\it st}-numbering~\cite{EvenT76} and planarity testing~\cite{HopcroftT74} among many others.

There are two versions of DFS studied in the literature. In the lexicographically smallest DFS or lex-DFS problem, when DFS looks for an unvisited vertex to visit in an adjacency list, it picks the “first” unvisited vertex where the “first” is with respect to the appearance order in the adjacency list. The resulting DFS tree will be unique. In contrast to lex-DFS, an algorithm that outputs some DFS numbering of a given graph, treats an adjacency list as a set, ignoring the order of appearance of vertices in it, and outputs a vertex ordering $Q$ such that there exists some adjacency ordering $R$ such that $Q$ is the DFS numbering with respect to $R$. We say that such a DFS algorithm performs general-DFS. In this work, we focus only on lex-DFS, thus, given a source vertex, the DFS tree is always unique. Given the lex-DFS tree, the non-tree edges of a given directed graph can be classified into four categories as follows. An edge directed from a vertex to its ancestor in the tree is called a back edge. Similarly, an edge directed from a vertex to its descendant in the tree is called a forward edge. Further, an edge directed from right to left in the DFS tree is called a cross edge. The remaining edges directed from left to right in the tree are called anti-cross edges. In the undirected graphs, there are no cross edges. Note that, we can store the complete DFS tree explicitly using $O(n \lg n)$ bits by storing pointers between nodes. In what follows, we formally define the problem which we call the \DFS\ problem.

\begin{center}
\fbox{\begin{minipage}{11cm}
\noindent{\DFS\ problem}\\
\noindent {\em Input}: A directed or undirected graph $G=(V,E)$ where $|V|=n$, $|E|=m$, and a source vertex $v_s$, preprocess $G$ and answer the following queries with respect to the DFS tree $T$ rooted at $v_s$:
\begin{enumerate}
\item Given any pair of vertices $v_i$ and $v_j$,
\begin{enumerate}
\item Who is visited first in the DFS traversal of $G$?
\item Is $v_i$ an ancestor of $v_j$ in $T$?
\end{enumerate}
\item Given $v_i$, 
\begin{enumerate}
\item Return the parent of $v_i$ in $T$.
\item Return the number of children (if any) of $v_i$ in $T$.
\item Enumerate all the children (if any) of $v_i$ in $T$.
\item Return the DFI of $v_i$.
\end{enumerate}
\item Enumerate the order in which vertices of $G$ are visited in the DFS.
\item Given $1\leq i \leq n$, return the vertex with DFI $i$.
\end{enumerate}
\end{minipage}}
\end{center} 

We study the \DFS\ problem in two well-known models: the {\it indexing} and {\it encoding} models~\cite{Navarro}. In the indexing model, we wish to build an index {\it ind} after preprocessing the input graph $G$ such that queries can be answered using both {\it ind} and $G$ whereas in the encoding model, we seek to build a data structure {\it encod} after preprocessing the input graph $G$ such that queries have to be answered using {\it encod} only. Typically the parameters of interest are (i) query time, (ii) space consumed (in bits) by {\it ind} and {\it encod} resp. and (iii) the preprocessing time and space. We address all these issues in our paper for the \DFS\ problem, assuming our computational model is a Random-Access-Machine with constant time operations on $O(\lg n)$-bit words. In both models, it is not hard to see that using $O(n \lg n)$ bits, we can answer all the queries of the \DFS\ problem in the optimal $O(1)$ time except the query of 3 which takes $O(n)$ time. 
Our main objective here is to beat this trivial $O(n \lg n)$ bit space bound without compromising too much on the query time. 

The motivation for studying this question mainly stems from the rise of the ``big data'' phenomenon and its implications. To illustrate, the rate at which we store data is increasing even faster than the speed and capacity of computing hardware. Thus, if we want to use the stored data efficiently, we need to represent it in sophisticated ways. Many applications dealing with huge data structures can benefit from keeping them in compressed form. Compression has many advantages: it can allow a representation to fit in main memory rather than swapping out to disk, and it improves cache performance since it allows more data to fit into the cache. However, such a data structure is only handy if it allows the application to perform fast queries to the data, and this is the direction we want to explore for the DFS tree. More specifically, we are interested in representing the DFS tree of a given graph compactly while supporting all the queries mentioned above efficiently.

\subsection{Representation of the Input Graph}
We assume that the input graphs $G=(V,E)$ are represented using the {\it adjacency array} format, i.e., $G$ is given by an array of length $|V|$ where the $i$-th entry stores a pointer to an array that stores all the neighbors of the $i$-th vertex. For the directed graphs, we assume that the input representation has both in/out adjacency array for all the vertices i.e., for directed graphs, every vertex $v$ has access to two arrays, one array is for all the in-neighbors of $v$ and the other array is for all the out-neighbors of $v$. This form of input graph representation has now become somewhat standard and was recently used in plenty of other works~\cite{Banerjee2018,Cha_thesis,CTW,Chakraborty00S18,ChakrabortyRS17,ChakrabortyS19}. Throughout this paper, we call a graph sparse when $m=O(n)$, and dense otherwise (i.e., $m=\omega(n)$).

\subsection{Our Main Results and Organization of the Paper}
We start by mentioning some preliminary results that will be used throughout the paper in Section~\ref{prelim}. Section~\ref{main_algo} contains the description of our main index for solving the \DFS\ problem in the indexing model.  Our main results here can be summarized as follows,


\begin{theorem}\label{sparsecase}
In the indexing model, given any sparse (dense resp.) undirected or directed graph $G$, there exists an $O(m+n)$ time and $O(n \lg n)$ bits preprocessing algorithm which outputs a data structure of size $O(n)$ ($O(n \lg (m/n))$ resp.) bits, using which the queries 1(a), 1(b), 2(d) and 4 can be reported in $O(\lg n)$ time, 2(a) and 2(b) 
in $O(1)$ time, 2(c) in time proportional to the number of solutions, and finally 3 can be solved in $O(n)$ time resp. for the \DFS\ problem.
\end{theorem} 

We want to emphasize that obtaining better results for sparse graphs is not only interesting from theoretical perspective
but also from practical point of view as these graphs do appear very frequently in most of the realistic network scenario in real world applications, e.g., Road networks and the Internet.
%

In Section~\ref{encoding}, we provide the detailed proof of our index in the encoding model. This contains a space lower bound for any index for the \DFS\ problem, followed by an index whose size asymptotically matches the lower bound and has efficient query time. We summarize our main results~below.

\begin{theorem}\label{encoding_lower_bound}
In the encoding model, the size of any data structure for the \DFS\ problem must be $\Omega(n \lg n)$ bits. On the other hand, given any (un)directed graph, there exists an $O(m+n)$ time and $O(n \lg n)$ bits preprocessing scheme that outputs an index of size $(1+\epsilon)n \lg n+2n+o(n)$ bits (for any constant $\epsilon > 0$), using which the queries 1(a), 1(b), 2(a), 2(b), 2(d) can be reported in $O(1)$ time, 2(c) in time proportional to the number of solutions, 3 in $O(n/\epsilon)$ time, and finally 4 in $O(1/\epsilon)$ time resp. for the \DFS\ problem in this setting.
\end{theorem}


Building on all these aforementioned results, we also show a host of applications of our techniques in designing indices for other fundamental graph problems in Appendix~\ref{appendix}. 
Finally, we conclude in Section~\ref{conclude} with some open problems and possible future directions to explore further.

{\bf Remark.} At this point we want to emphasize that our results are more general, i.e., they can be extended to store any arbitrary labeled tree (arising from some underlying graph) along with the mechanism for fast querying. This method is very useful as many graph algorithms (like shortest path, minimum spanning tree, biconnectivity etc) induce a tree structure which is used subsequently during the execution of the algorithm. Hence, we can use our technique to store and query those trees compactly as well as efficiently. Thus, we also believe that our algorithm may find many other potential interesting applications. However, we chose to provide all the details in terms of DFS as DFS is very widely popular graph traversal technique and is used as the backbone for multiple fundamental algorithms, yet there is no explicit indexing scheme for storing DFS tree compactly. In Appendix~\ref{app}, 
we show how one can extend these techniques to design indexing schemes for a variety of other classical and fundamental graph problems.

\subsection{Related Works}
There already exists a large body of work concerning compactly representing various specific classes of graphs, for example planar, constant genus graphs etc~\cite{Acan,BlandfordBK03,FerresSG0N17,MunroN16,MunroR01,Navarro,YamanakaN10}. All of these works are able to store an $n$-vertex unlabeled planar graph in $O(n)$ bits, and some of them even allow for $O(1)$-time neighbor queries. Generally what is meant by unlabeled is that the algorithm is free to choose an ordering on the vertices (integer labels from $1$ to $n$). Our setting here is slightly different as we work with graphs whose vertices are labeled, and matches closely with~\cite{BarbayAHM07}. Also we want to support more complex queries whereas the previous works only focused on adjacency queries mostly. Even though DFS being such a widely known method, and having many applications, to the best of our knowledge, we are not aware of any previous work focusing on compactly representing the DFS tree with efficient query support.

\section{Preliminaries} \label{prelim}
{\bf Rank-Select.} We make use of the following theorem:
\begin{theorem}\cite{Clark96}
 \label{staticrs}
We can store a bitstring $B$ of length $n$ with additional $o(n)$ bits such that rank and select operations (defined below) can be supported in $O(1)$ time. Such a structure can also be constructed from the given bitstring in $O(n)$ time and space.
\end{theorem}

For any $ a\in \{0,1\}$, the rank and select operations are defined as follows :
\begin{itemize}
 \item $rank_a(B,i)$ = the number of occurrences of $a$
 in $B[1,i]$, for $1\leq i\leq n$;
 \item $select_a(B,i)$ = the position in $B$ of the $i$-th occurrence of $a$, for $1\leq i\leq n$.
\end{itemize}

When the bitvector $B$ is sparse, the space overhead of $o(n)$ bits can be avoided by using the following theorem, which will also be used later in our paper. 
\begin{theorem}\cite{Navarro}
\label{sparse}
We can store a bitstring $B$ of length $n$ with $m$ $1$s using $m\lg(n/m)+O(m)$ bits such that $select_1(B,1)$ can be supported in $O(1)$ time, $select_0(B,1)$ in $O(\lg m)$ time, and both the rank queries ($rank_1(B,i)$ and $rank_0(B,i)$) can be supported in $O(\text{min}(\lg m,\lg n/m))$ time. Such a structure can also be constructed from $B$ in $O(n)$ time and space.
\end{theorem}

{\bf Permutation.} We also use the following theorem:
\begin{theorem}\cite{MunroRRR12}
\label{perm}
A permutation $\pi$ of length $n$ can be represented using $(1+\epsilon)n \lg n$ bits so that $\pi(i)$ is answered in $O(1)$ time and $\pi^{-1}$ in time $O(1/\epsilon)$ for any constant $\epsilon > 0$. Such a representation can be constructed using $O(n)$ time and space.
\end{theorem}

{\bf Succinct Tree Representation.} We need following result from~\cite{FarzanM14}.
\begin{theorem}~\cite{FarzanM14}
\label{succ_tree}
There exists a data structure to succinctly encode an ordered tree with $n$ nodes using $2n+o(n)$ bits such that, given a node $v$, (a) child($v$,$i$): $i$-th child of $v$, (b) degree($v$): number of children of $v$, (c) depth($v$): depth of $v$, (d) $select_{pre}$($v$): position of $v$ in preorder, (e) $LA(v,i)$: ancestor of $v$ at level $i$ can be supported in $O(1)$ time among many others. Such a structure can also be constructed in $O(n)$ time and space.
\end{theorem}

\section{Algorithms in the Indexing Model}\label{main_algo}
In this section, we provide the main algorithmic ideas needed for the solution of the \DFS\ problem in the indexing model. We start by describing the preprocessing procedure which is followed by the query algorithms.
\subsection{Preprocessing Step}
We first describe our algorithms for undirected graphs, and later mention the modifications required for the case of directed graphs. The preprocessing step of the algorithm is divided into two parts. In the first part, we perform a DFS of the input graph $G$ along with storing some necessary data structures. In the second step, we perform a partition of the DFS tree of $G$ using the well-known ``tree covering technique'' of the succinct data structures world~\cite{FarzanM11}, and also store some auxiliary data structures. Later, in the final step of our algorithm, we show how to use these data structures to answer the required queries. In what follows, we describe each step in detail.

{\bf Step 1: Creating Parent-Child Array using Unary Degree Sequence Array.} The main idea of this step is to perform a DFS traversal of $G$ and store in a {\it compact way} the parent-child relationship of the DFS tree $T$. The way we achieve this is by using three bitvectors of length $O(m+n)$ bits.
Recall that, our input graphs $G=(V,E)$ are represented using the standard adjacency array. Central to our preprocessing algorithm is an encoding of the degrees of the vertices in unary. As usual, let $V =\{v_1, v_2, \cdots, v_n\}$ be the vertex set of $G$. The unary degree sequence encoding $D$ of the undirected graph $G$ has $n$ $1$s to represent the $n$ vertices and each $1$ is followed by a number of $0$s equal to its degree. Moreover, if $d$ is the degree of vertex $v_i$, then $d$ $0$s following the $i$-th $1$ in the $D$ array corresponds to $d$ neighbors of $v_i$ (or equivalently the edges from $v_i$ to the $d$ neighbors of $v_i$) in the same order as in the adjacency array of $v_i$. Clearly $D$ uses $n+2m$ bits and can be obtained from the neighbors of each vertex in $O(m+n)$ time.  Now using {\it rank/select} queries of Theorem~\ref{staticrs} in Section~\ref{prelim}, the $j$-th outgoing edge of vertex $v_i$ can be identified with the position $p = select_1(D,i)+j$ of $D$ ($1 \le j \le degree(v_i)$ where $degree(v_i)$ denotes the degree of the vertex $v_i$). From a position $p$, we can obtain an endpoint of the corresponding edge by $i = rank_1(D, p)$, and the other endpoint is the $j$-th neighbor of $v_i$ where $j = p - select_1(D, i)$.

We also use two bitvectors $E, P$ of the same length where every bit is initialized to $0$, and the bits in $E, P$ are in one-to-one correspondence with bits in $D$. The bitvector $E$ will be used to mark the tree edges of the DFS tree $T$, and the bitvector $P$ to mark the unique parent of every vertex in $T$.
The marking is carried out while performing a DFS of $G$ in the preprocessing step. I.e., if $(v_i, v_j)$ is an edge in the DFS tree where $v_i$ is the parent of $v_j$, and suppose $k$ is the index of the edge $(v_i, v_j)$ in $D$, then the corresponding location in $E$ is marked as $1$ during DFS. At the same time, we scan the adjacency array of $v_j$ to find the position of $v_i$ (as $G$ is undirected, there will be two entries for each edge in the adjacency array), and suppose $t$ is the index of the edge $(v_j,v_i)$ in $D$, then the corresponding location in $P$ is marked as $1$ during DFS. Thus, assuming $G$ is a connected graph, once DFS finishes traversing $G$, the number of ones in $E$ is exactly the number of tree edges (which is $n-1$) and the number of ones in $P$ will be $n-1$ as root does not have any parent. 

The parent of $v_i$ in $T$ is computed in $O(1)$ time
as follows.
Let $v_r$ be the root of $T$.  Then if $i > r$
(resp. $i<r$), the marked bit representing the parent of
$v_i$ is the $(i-1)$-st (resp. $i$-th) $1$ in $P$.
Let $p = select_1(P, i-1)$ (resp. $p = select_1(P, i)$)
and $j = p - select_1(D, i)$.  Then the parent of $v_i$
is the $j$-th neighbor of $v_i$.

We use another bitvector $D_T$ of length $2n$,
which encodes the degree of each vertex in $T$
by unary sequences.  Then the degree of vertex $v_i$ in $T$
is $select_1(D_T, i+1)-select_1(D_T, i)-1$, and
$j$-th child of $v_i$ in $T$ is $p$-th neighbor of $v_i$
in $G$ where $p = select_1(E, select_1(D_T, i-1)+j)-select_1(D, i)$.  These are computed in constant time.

Note that, the classical linear time implementation of DFS~\cite{CLRS} uses a stack (which could grow to $O(n \lg n)$ bits) and a color array (of size $O(n)$ bits). Thus, the procedure takes $O(m+n)$ time and $O(n \lg n)$ bits overall. First, we argue that using the same linear time, we can also create bitvectors $D, E$ and $P$ and fill up them correctly. It's easy to see that creating $D$ as well as initializing $E$ and $P$ to all zero takes $O(m+n)$ time. All it remains is to show, how one can fill up $E$ and $P$ while performing DFS. For this purpose, we build the data structures to support the constant time {\it rank/select} query (of Theorem~\ref{staticrs}) on $D$ (and on $E, P$ as well, the reason will be clear in the query step) and use the result of the select query to mark the tree edges on $E$ (as they are in one-to-one correspondence). To illustrate, suppose, while traversing from $v_i$, DFS discovers the edge $(v_i, v_j)$ as a tree edge in $T$ where $v_i$ is the parent of $v_j$, and suppose $v_j$ is the $c$-th neighbor in $v_i$'s adjacency array, then we find the index of the $c$-th zero after $i$-th one in $D$ (using select query), and the corresponding index is marked as $1$ in the $E$ array. This takes $O(1)$ time for each tree edge marking. After this, we mark the index in $P$ as $1$ corresponding to the edge $(v_j,v_i)$ to denote that $v_i$ is the parent of $v_j$. Thus, marking parent takes $O(degree({v_j}))$ time for the vertex $v_j$. Note that, all of this happens along with the classical stack-based DFS implementation. Thus overall it takes $O(m+n)$ time, and space required to store all these arrays is $O(m+n)$ bits. We refer to the bitvector $D$ as the unary degree sequence array, $E$ as the child array, and $P$ the parent array. These three arrays are stored and used for the query step of our algorithm. Thus, we obtain the following lemma.

\begin{lemma}
Given an undirected graph $G$, there exists an $O(m+n)$ time and $O(n \lg n)$ bits preprocessing algorithm to construct the unary degree sequence array, parent and child arrays for $G$, each of which takes $O(m+n)$ bits of space.
\end{lemma}

{\bf Step 2: Decomposing the DFS tree by the Tree Covering Technique.} The main idea of this step is to perform a decomposition of the DFS tree, and along with storing some crucial informations which will be very useful for navigating the tree during the query step of our algorithm. For this purpose, we use the well-known tree covering technique in the context 
of succinct representation of rooted ordered trees. The high level idea is to decompose the tree into subtrees called 
{\it minitrees}, and further decompose the minitrees into yet smaller subtrees called {\it microtrees}. The microtrees are small enough to be stored in a compact table. The root of a 
minitree can be shared by several other minitrees. To represent 
the tree, we only have to represent the connections and links between the subtrees. One such tree decomposition method was given by Farzan and Munro~\cite{FarzanM11} where each minitree has at most one node, other than the root of the minitree, that is connected to the root of another minitree. This guarantees that in each minitree, there exists at most one non-root node which is 
connected to (the root of) another minitree. We use this decomposition in our algorithms, and the main result of Farzan et al.~\cite{FarzanM11} is summarized in the following theorem:

\begin{theorem}[\cite{FarzanM11}]\label{thm:tree-decomposition}
For any parameter $L \ge 1$, a rooted ordered tree with $n$ nodes can be decomposed into $\Theta(n/L)$ minitrees of size at most $2L$ which are pairwise disjoint aside from the minitree 
roots. Furthermore, aside from edges stemming from the minitree root, there is at most one edge 
leaving a node of a minitree to its child in another minitree. The decomposition can be performed 
in linear time using linear words of space.
\end{theorem}

\begin{figure}[bt]
\begin{center}
 \includegraphics[scale=.45, keepaspectratio=true]{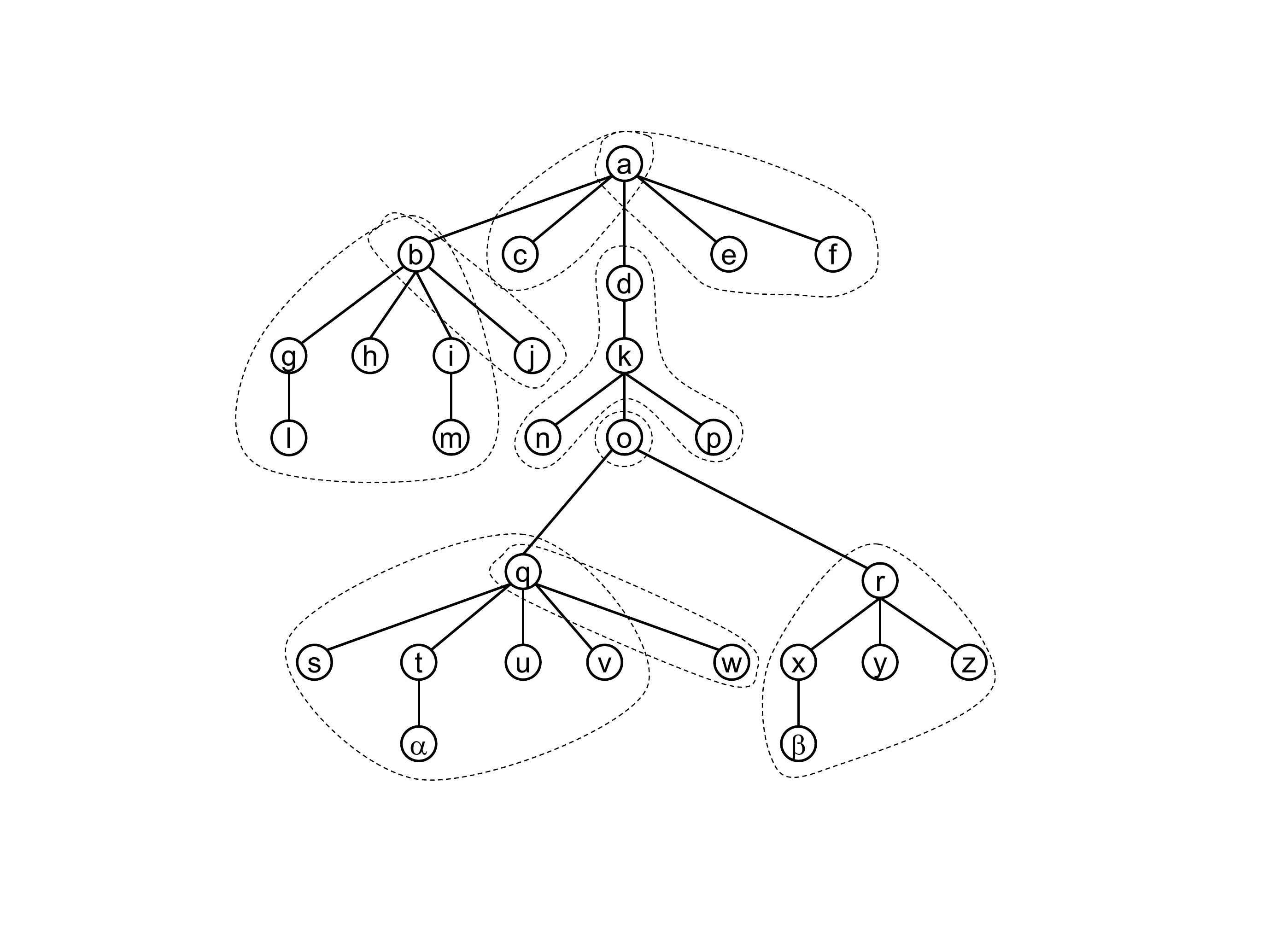}
\end{center}
\caption{An example of Tree Covering technique with $L=5$. 
Each closed region formed by the dotted lines represents a minitree. Here each minitree has at most one `child' minitree (other than the minitrees that share its root).}\label{fig:1}
\end{figure}

See Figure~\ref{fig:1} for an illustration. For the purpose of our algorithms, we apply Theorem~\ref{thm:tree-decomposition} with $L =\lg n$ on the DFS tree $T$ of $G$. For this parameter $L$, since the number of minitrees is only $O(n/\lg n)$, we can represent the structure of the minitrees within the original tree (i.e., how the minitrees are connected with each other) 
using $O(n)$ bits by simply storing both way pointers (so that we can traverse easily) between the roots of the minitrees. We refer to this as the {\it skeleton} $S$ of the DFS tree $T$. See Figure~\ref{fig:2} for a demonstration of Figure~\ref{fig:1}'s skeleton. The decomposition algorithm of~\cite{FarzanM11} also ensures that each minitree has at most one `child' minitree (other than the minitrees that share its root) in this structure. We use this property crucially later. 

\begin{figure}[bt]
\begin{center}
 \includegraphics[scale=.7, keepaspectratio=true]{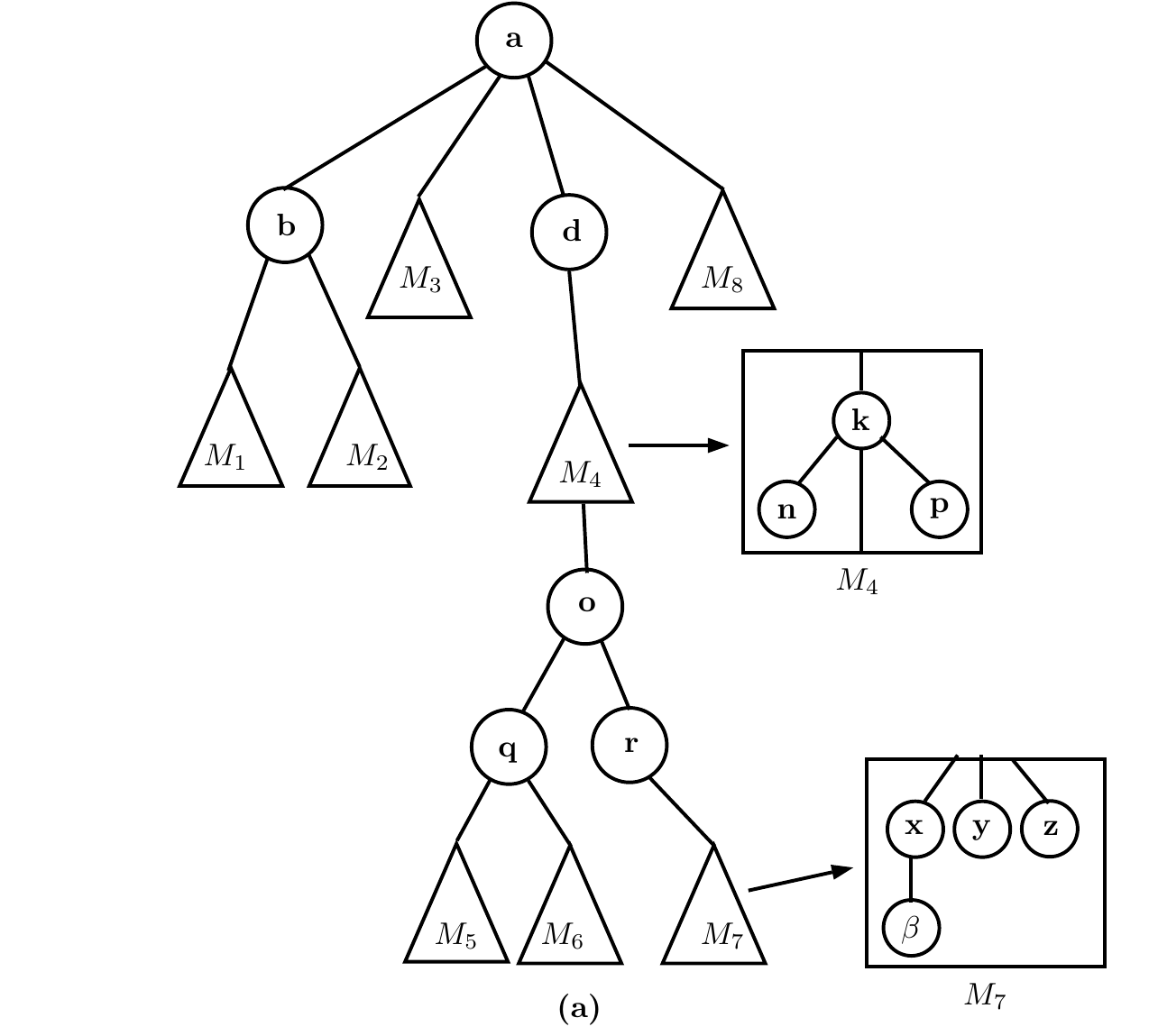}
\end{center}
\caption{(a) A rough sketch of the skeleton of the tree decomposition shown in Figure~\ref{fig:1}. In this diagram, the triangles represent the minitrees along with the roots of the minitrees are marked inside the circle. For example, the minitrees $M_1$ and $M_2$ share the same root $b$. Also the node $o$ is a minitree on its own. Strictly speaking, the skeleton will not have the traingles, rather it just contains the pointers between the roots of the minitrees (i.e., circles in this diagram). But we put this diagram for better visual description of the compact representation of the previous diagram.}\label{fig:2}
\end{figure}

In what follows, we explain how we compactly represent the minitree structure, and we refer to this compact representation obtained using this tree covering (TC) approach as the TC representation of the DFS tree. Towards this, first observe that every minitree root has unique first child and last child inside the minitree. In some cases, both are the same (see the minitree rooted at node $d$ of Figure~\ref{fig:1}), and in some cases, both are absent (see the minitree rooted at node $o$ of Figure~\ref{fig:1}). Thus, if we specify these two quantities, we can uniquely identify the root of the minitree (along with the exact portion of the nodes which are children of the root of this minitree and also belong to the same minitree as the first and last child of the root) even though the root is shared between multiple minitrees. We use this idea crucially in the design of the TC representation of the DFS tree. 

We mark in a bitvector $R$ of size $n$ all the nodes which are the last child of a minitree root inside a minitree. Note that, there are 
$O(n/\lg n)$ such nodes which are marked as $1$ in $R$. In the case of a minitree root not having any children, we mark the minitree root itself as $1$ in $R$. We also build the data structure to support $O(1)$ time {\it rank/select} queries on $R$ using Theorem~\ref{staticrs}. Next, we create an array $C$ where each of the $O(n/\lg n)$ entries are $O(\lg n)$ bits long, thus overall it takes $O(n)$ bits. Basically, each entry of $C$ stores some informations regarding the minitree for which the last child of the minitree root is marked $1$ in $R$. More specifically, For a typical node, say $v_i$, which is the last child of some minitree, we have $R[i]=1$, and $C[j]$ (where $j=rank_1(R,i)$) comprises of the following six informations (some of which could be empty), (i) label of the minitree root, say $v_r$, for which $v_i$ is the last child inside the minitree, (ii) location of the first child, say $v_j$, of $v_r$ inside the minitree in the adjacency array of $v_r$, (iii) DFI of $v_j$, (iv) the edge $(v_c,v_d)$ (if any) that goes out of the minitree, (v) the size of the subtree rooted at $v_c$ in the DFS tree, (vi) depth of $v_r$ in the DFS tree. The tree decomposition method ensures that a minitree has at most one edge $(v_c,v_d)$, where $v_c$ is a non-root node of minitree and $v_d$ is a root of a different minitree, that goes out of the minitree. We also mark in a bitvector $Z$ of size $n$ bits all such vertices like $v_c$ (also note, there could be $O(n/\lg n)$ such vertices). We mark in a bitvector $L$ all the vertices which are the rightmost leaves of every minitree. Note that these vertices (there are, again, $O(n/\lg n)$ of them) have the highest DFI inside the minitree. In another bitvector $A$, we mark all the roots of the minitrees as $1$, and build rank/select structure on top of $A$. Correspondingly, the $F$ array will store the DFI of the roots so that we can retrieve them in constant time. More specifically, for a minitree root $v_r$, $A[r]=1$ and $F[j]$ ($j=rank_1(A,r)$) will store the DFI of $v_r$. Next we build the $O(1)$ time level ancestor data sturcture, say $LA$, on the $O(n/\lg n)$ minitree roots (i.e., on the skeleton structure) using~\cite{BenderF04}. Thus, here, $LA$ takes $O(n)$ bits and $O(n)$ time. As a root of the minitree is shared between multiple minitrees, from each node $v_i$ of the skeleton $S$ (where $v_i$ is a root of a minitree), we store pointers to all the minitrees (in $C$ array) which has $v_i$ as their root. Overall these pointers also consume $O(n)$ bits. This completes the description of the TC representation. Note that, the creation of the skeleton and the TC representation for $T$ can be done in $O(n)$ time using $O(n \lg n)$ bits (using Theorem~\ref{thm:tree-decomposition}) after the DFS (which takes $O(m+n)$ time and $O(n \lg n)$ bits). Hence, we obtain the following,

\begin{lemma}
Given an undirected graph $G$, there exists an $O(m+n)$ time and $O(n \lg n)$ bits preprocessing algorithm to construct the skeleton $S$ and the TC representation of the DFS tree $T$ of $G$, each of which occupies $O(n)$ bits.
\end{lemma}

First observe that, the outputs of the previous step are the unary degree sequence array ($D$), parent array ($E$), child  array ($P$), the $D_T$ array, TC representation of $T$ (this includes $R$, $C$, $Z$, $A$, $F$ and $L$) along with the skeleton $S$ with pointers to $C$, and finally the $LA$ structure on $S$. The arrays $D$, $E$, and $P$ take $O(m+n)$ bits, and the others take $O(n)$ bits. Now we show how to efficiently solve the \DFS\ problem using these structures. 

{\bf Query Algorithms.} Given $v_i$, to answer 2(a) in $O(1)$ time, we do the following. If $v_i$ is the root of the DFS tree, we return null. 
Otherwise, we can compute the answer by using
only $select_1$ queries on $P$ and $D$, as described previously.

To answer 2(b) in $O(1)$ time, 
we use $select_1$ queries on the bitvector $D_T$.

To answer 2(c), 
we first compute the number of children of $v_i$ in $T$
using the query 2(b).  Then $j$-th child is obtained
in constant time as described above.

Note that, the queries 2(a), 2(b) and 2(c) can be answered using only $D$, $E$, $P$ and $D_T$ arrays. Before explaining the algorithms for the rest of the queries, we first prove
the following very crucial lemma.
\begin{lemma}\label{reconstruct}
Given any query node $v_i$ which is not a root of a minitree, we can reconstruct the minitree $M$ containing $v_i$ in time proportional to the size of $M$ along with the DFIs of all the nodes inside $M$. In the same amount of time, we can also retrieve the root node of $M$.
\end{lemma}

\begin{proof}
First note that if a node $v$ belongs to the minitree $M$,
its children in $T$ also belong to $M$, except for the
following two cases.  The first case is that $v$ is the
root $v_r$ of $M$ and the second case is that $v$ is $v_c$
of $M$.  In the first case, as we have stored the location (in the adjacency array) of the first child, say $v_j$, of $v_r$ inside $M$ in the $C$ array, we can enumerate all the children of $v_r$ in $M$
in constant time for each until we hit the rightmost child
of $v_r$ in $M$, which is stored in $R$.
In the second case, we can also enumerate all the children
of $v_c$ in $T$ in constant time for each and discard $v_d$.
For other vertices in $M$, we can enumerate children
using constant time for each output.  Note also that
going to the parent can be performed in constant time.

The algorithm to achieve the claim can be broken down into three steps. In the first step, given $v_i$, we launch a DFS starting from $v_i$, and continue till we retrieve the rightmost leaf, say $v_j$, of the minitree $M$. In second step, we follow the path in $T$ (by going to $v_j$'s parent, then its parent and so on) till we reach the rightmost child, say $v_k$, of the root, say $v_r$, inside $M$ by using the query algorithm to find the parent repeatedly. In the third and final step, we use the $C$ array, by using $v_k$, to extract all the informations needed to reconstruct the full minitree by performing another round of DFS. We provide the details below.

To perform the first step, we only need to use the parent and child related queries, whose execution we already showed previously. Note that, as we have stored the information (in $Z$ array) regarding the only edge that goes out of the minitree, we never incorrectly go out of $M$. Also we can verify if we have reached the unique node $v_j$ which is the rightmost leaf of $M$ from the $L$ array. Once we reach $v_j$, it's easy to see that $v_k$ has to be an ancestor of $v_j$ (note that $v_k$ and $v_j$ could be same in some cases). Thus, we can reach $v_k$ from $v_j$ by repeatedly using the parent query algorithm, and this completes the second step. Finally, once we reach $v_k$, we use the informations in $C[j]$ (where $j=rank_1(R,k)$) to retrieve the root $v_r$ of $M$ and other informations. Then we carry out a DFS from $v_r$ by first going to the first child of $v_r$ inside $M$ (retrieved from $C$), then its first child and so on till we fully reconstruct $M$. This step also requires repeated invoking of parent and child query only.

In order to retrieve the DFIs of the nodes inside $M$, observe that, if $M$ doesn't have any child minitree (i.e., no edge is going out of $M$), then while doing the final DFS from $v_r$, we can easily compute the DFIs of all the nodes inside $M$. Otherwise, assume the edge $(v_c,v_d)$ goes out of $M$ where $v_c$ belongs to $M$, then the DFI of next node inside $M$ can be calculated by adding the size of the subtree rooted at $v_c$ in $T$ (which is stored in $C$) to the DFI of $v_c$. It is clear that all of these procedures can be performed in the time proportional to the size of $M$, which is $O(\lg n)$ here. This completes the description of the proof.
\end{proof}

As a corollary of the previous lemma, it is easy to see that the query of 2(d) can be reported in $O(\lg n)$ time for any node $v_i$ which is not a root of the minitree. Otherwise, it can be done in $O(1)$ time by reporting the value stored in $F[j]$ where $j=rank_1(A,i)$. 

To answer 1(a), first we invoke query algorithm of 2(d) for both $v_i$ and $v_j$ to retrieve their DFIs respectively, and then answer accordingly. Thus, this also takes $O(\lg n)$ time.

Answering 1(b) involves a few cases. In the first case, if both of them belong to the same minitree then we can figure out the answer by reconstructing the complete minitree. Secondly, suppose $v_i$ and $v_j$ are roots of the two separate minitrees, and their depths in $T$ are $x$ and $y$ respectively (depth can be obtained from $C$ array). Then using these values in $LA$ data structure, we can figure out the required answer. Finally, if both of these nodes belong to two separate minitrees but are not the roots of the minitrees, then first we retrieve the roots of those minitrees using Lemma~\ref{reconstruct}, then follow almost the same procedure as before to figure out the answer. Note that, in this case, it is enough to reconstruct the path from $v_c$ (of the minitree located near to the root) to the root of that minitree (for the case when one of the minitree root is an ancestor of the other minitree root) to figure out the answer of the query. Thus, overall, it takes $O(\lg n)$ time.

To return the query for 3, we do a standard DFS traversal on the skeleton $S$ and each time we visit a new node $v_i$ in $S$, we follow the pointer from $v_i$ in $S$ to the part of the $C$ array where the informations regarding the minitree rooted at $v_i$ is stored. Note that, $v_i$ might be shared between multiple minitrees, hence, we always start following these pointers from left to right. More specifically, if $v_i$ is the root of $p$ minitrees, we have $p$ pointers emanating from $v_i$, and going to $p$ different locations of $C$ array. As these pointers are stored from left to right order, which is the same order in DFS of all the minitrees that share the root $v_i$. Thus we follow the first pointer, and reach the specific portion of $C$, use Lemma~\ref{reconstruct} to generate the complete minitree along with the DFIs of the nodes. Then if this minitree has any child minitree, we go on to explore that and so on (by following the $(v_c,v_d)$ edge stored in that minitree). Once we finish all the descendant minitree of the first minitree rooted at $v_i$, we come back and start exploring the second minitree (by following the second pointer from $v_i$) and so on. Thus, we need to store these intermediate pointers, in stack, to know how much progress has been made in every node's (in skeleton) list. This procedure is continued until all the nodes of $S$ are exhausted. It is clear that this procedure takes $O(n)$ time as there are $O(n/\lg n)$ nodes in $S$ and for each node, we spend $O(\lg n)$ time. Also, we need $O(n)$ bits (as there could be $O(n/\lg n)$ pointers) of intermediate space for the execution of the DFS. 

To answer 4, first note that, in any minitree $M$, if there is no egde going out (i.e., no $(v_c,v_d)$ type edge), then the DFIs inside $M$ are consecutive, i.e., in general, first child of root inside $M$ has the smallest DFI and the rightmost leaf in $M$ has the maximum DFI, and the numbers are consecutive. Otherwise, DFIs are consecutive from the root of $M$ to the DFI of $v_c$, then there is a jump of DFI by the size of the subtree rooted at $v_c$ in the DFS tree, then it's consecutive DFI again until the rightmost leaf (which has the largest DFI inside $M$) of $M$. Thus, the range of DFIs of the vertices inside any arbitrary minitree $M$ can be broken into at most two disjoint consecutive intervals. We store these (at most $O(n/\lg n)$) intervals in an interval tree along with augmenting it with the last child of $M$ inside $M$. Now, given $i$, we first find the interval where $i$ belongs to from the tree and simultaneously retrieve the last child, say $v_a$, of the corresponding minitree, all using $O(\lg n)$ time. Then, we use the information from $R$ and $C$ array corresponding to $v_a$ to invoke Lemma~\ref{reconstruct}, and retrieve the desired vertex with DFI $i$ using $O(\lg n)$ overall time. This completes the description of the query algorithms for undirected graphs. 

We can handle directed graphs similarly except a few changes in the data structures. Recall that, for directed graphs, every vertex $v_i$ has access to its in-neighbors array as well as out-neighbors array, and additionally we create two unary degree sequence arrays (each of size $O(m+n)$ bits), $D_1$ for the out-neighbors and $D_2$ for the in-neighbors. We also have two separate arrays, say $E_1$ (having one-to-one map with $D_1$), for marking child of every node and $E_2$ (having one-to-one map with $D_2$) where parents are marked. It is easy to see that almost in a similar fashion as in the undirected case, we can correctly mark, for any node $v_i$, the children of $v_i$ in $E_1$ array and parent of $v_i$ in $E_2$ array using both the $D_1$ and $D_2$ arrays while performing DFS of $G$. The second preprocessing step doesn't require any changes for the  directed graphs. Now reporting queries also can be suitably modified to make use of these changes without affecting the asymptotic running time of the query algorithms. Basically the only change that takes place is as follows, whenever we need to find the parent of a node, now we need to use the in-neighbor array whereas finding children can be handled by consulting out-neighbor array along with the mapping with their respective unary degree sequence array. We omit the details. Thus, we obtain the following,

\begin{theorem}\label{main_theorem}
Given any undirected or directed graph $G$, there exists an $O(m+n)$ time and $O(n \lg n)$ bits preprocessing algorithm which outputs a data structure of size $O(m+n)$ bits, using which the queries 1(a), 1(b), 2(d) and 4 can be reported in $O(\lg n)$ time, 2(a) and 2(b) in $O(1)$ time, 2(c) can be answered in time proportional to the number of solutions, and finally 3 can be solved in $O(n)$ time respectively for the \DFS\ problem.
\end{theorem}

Note that, if the given input graph is sparse (i.e., $m=O(n)$), then both unary degree sequence array ($D$), and parent and child arrays ($E, P$) take $O(n)$ bits, and every other data structure anyway takes $O(n)$ in total, thus, we obtain the result mentioned in Theorem~\ref{sparsecase} for the case of sparse graphs. When the input graph is dense (i.e., $m=\omega(n)$), we compress the $D, E, P$ arrays using Theorem~\ref{sparse}. Note that we use only $select_1$ queries on compressed
arrays and thus query time complexity on the arrays
is still constant. Hence, we obtain the result of Theorem~\ref{sparsecase} for the dense graph case. It is worth mentioning that except the case for very dense graphs, our space bound {\it always} beats the space bound of the naive algorithm for every edge density in the full spectrum, albeit with super-constant query time. Thus, when the graph is sufficiently dense, it is better to use the standard solution which uses $O(n \lg n)$ bits with constant query time. This completes the description of our algorithms in the indexing model, and hence, the proof of Theorem~\ref{sparsecase}.

\section{Algorithms in the Encoding Model}\label{encoding}
Recall that in the encoding model, we seek to build a data structure {\it encod}
after preprocessing the input graph $G$ such that queries have to be answered using {\it encod} only, {\it without} accessing $G$. To this end, we first provide a lower bound for the space needed for {\it encod} to answer queries of the \DFS\ problem.

\subsection{Space lower-bound}
Observe that, in order to correctly answer the queries, the data structure {\it encod} must contain the information regarding the topology of the DFS tree $T$ of the graph $G$ along with the labels of the vertices of $T$ as, unlike the indexing model, we don't have the access to $G$ during the query time. It's easy to see that we need $\Omega(n \lg n)$ bits to store the vertex labels mappings. In what follows, we give a proof for the space needed to store the topology of the DFS tree by counting the number of such trees in any arbitrary graph $G$.

\begin{lemma}
For a graph with $n$ vertices and $m$ edges, the size of data structures
for storing the topology of the DFS trees is $\Omega\left(n \lg \frac{m}{n}\right)$ bits.
\end{lemma}

\begin{proof}
Let us consider the following graph $G$ with $n$ vertices and $m$ edges ($m < n^2/2$).
It has a vertex $r$, $k = m/n$ vertices $u_1,\ldots,u_k$,
and $n-k-1$ vertices $v_1,\ldots, v_{n-k-1}$.
The vertex $r$ is connected to all $u_i$, and each $v_j$ is also connected to
all $u_i$.  To construct a spanning tree of $G$, we choose one edge among all
$k$ edges connected to each $v_j$.  Then the number of different spanning trees of $G$ is at least $k^{n-k-1}$, and for all different spanning trees the set of DFI's are different. Therefore the size of data structure must be
at least $\lg k^{n-k-1}$ bits, which is~$\Omega\left(n \lg \frac{m}{n}\right)$.
\end{proof}
Thus, the space lower bound for {\it encod} is $\Omega(\textit{max}\{n \lg \frac{m}{n}, n \lg n\})$ bits, which is $\Omega(n \lg n)$ bits as mentioned in Theorem~\ref{encoding_lower_bound}. In what follows, we complement the above claim by providing a simple indexing structure which asymptotically matches this lower bound.

\subsection{Upper-bound}
{\bf Preprocessing Step.} Our index for the \DFS\ problem consists of two components which we prepare during the preprocessing step. In the first component, we store, for every vertex $v_i$, DFI($v_i$) as permutation using the structure of Theorem~\ref{perm} of Section~\ref{prelim}. Secondly, we encode the DFS tree succinctly using the structure of Theorem~\ref{succ_tree} of Section~\ref{prelim}.

{\bf Query Algorithm.} We answer the queries using the two above mentioned structures as follows. To answer 2(d), we just use $\pi(i)$. 
Similarly, 4 can be answered by invoking $\pi^{-1}(i)$. We report $v_i$ (resp. $v_j$) as the answer for query 1(a) if $\pi(i) < \pi(j)$ (resp. $\pi(i) > \pi(j)$). We enumerate the vertex ordering as traversed in the DFS order by invoking $\pi^{-1}(1)$, then $\pi^{-1}(2)$, and so on till $\pi^{-1}(n)$. We answer 1(b) in affirmative by checking if $LA(v_j,depth(v_i))$ matches with $v_i$, otherwise no. To answer 2(a), we return $LA(v_i,depth(v_i)-1)$. We return the answer of 2(b) by using the query $degree(v_i)$. Finally, we enumerate the children of a node $v_i$ as requested in query 2(c) by using the query $child(v_i,1)$ till $child(v_i,degree(v_i))$. Hence we obtain the results mentioned in Theorem~\ref{encoding_lower_bound}.

\section{Conclusion}\label{conclude}
In this paper, we provided procedures for compactly storing the DFS tree for any graph with efficiently supporting various queries in the indexing and encoding models, and showed how to extend these techniques to design indexing schemes for other fundamental and basic graph problems. With some work, our algorithm can be extended for indexing BFS tree (and other graph search tree also) as well while supporting similar types of queries. Also, as mentioned previously, our results are more general, and can be used in other situations as well.

This work opens up many possible future directions to explore. Can we further improve the query time while keeping the space bound same in the indexing model? Can we prove a space lower bound in the indexing model? Can we design compact data structures for indexing problems like maximum flow?
Finally, we conclude by remarking that using 
\cite{Banerjee2018,ChakrabortyRS17}, we can improve the preprocessing space of our algorithms
to $O(n)$ bits (from $O(n \lg n)$ bits) with marginal 
increment in the preprocessing time.

\bibliographystyle{plain}
\bibliography{short-bib}

\newpage
\appendix
\section{Appendix} 
\label{appendix}

\subsection{Applications}\label{app}
In this section we discuss how to design indexing structure for various graph problems using the techniques we developed earlier. More specifically, we develop indexing structures for shortest path, undirected connectivity, bi-connectivity, $2$-edge connectivity and strong connectivity in the indexing model. With some effort, they can also be easily extended for the encoding model as well. We start with indexing shortest paths first.

\subsection{Indexing Shortest Path}
In what follows, we assume that, for the weighted graphs, the adjacency array also has the weights along with the neighbors.  We start by defining the following problem which we call the \SHORT\ problem.

\begin{center}
\fbox{\begin{minipage}{11cm}
\noindent{\SHORT\ problem}\\
\noindent {\em Input}: An undirected or a directed graph $G=(V,E)$ where $|V|=n$, $|E|=m$ and non-negative $O(\lg n)$ bit edge weights, and a source vertex $v_s$. Preprocess $G$ and answer the following queries:\\
\noindent{\em Queries}: 
\begin{enumerate}
\item Given any vertex $v_i$, 
\begin{enumerate}
\item Return the length of the shortest path between $v_s$ and $v_i$.
\item Return a shortest path from $v_s$ to $v_i$.
\end{enumerate}
\end{enumerate}
\end{minipage}}
\end{center}

Note that, given a source vertex $v_s$, a shortest path tree in $G$ rooted at $v_s$ is a tree that is the union, over all $v_i \in V$, of a shortest path in $G$ from $v_s$ to $v_i$, and this tree can be computed by running Dijkstra's algorithm~\cite{CLRS} on $G$. Even though we described in the previous section how one can encode the DFS tree of any graph compactly, observe that, the method (for storing the tree and querying as well) works for any arbitrary rooted tree in general. Thus, in the preprecessing step, we run the classical Dijkstra's algorithm which takes $O(m+n \lg n)$ time and $O(n \lg n)$ bits of space, and computes the shortest path tree $T_s$. We run our preprocessing algorithm on $T_s$ to store it compactly in exactly the same way as we did for the DFS tree with only one extra piece of information. With the every entry in the $F$ array (where we store the DFIs of the roots of the minitrees), we also the store the length of the shortest path from the source $v_s$ to the roots of the minitrees. Then, to answer the query of 1(a), using a similar procedure of Lemma~\ref{reconstruct}, we first reach from $v_i$ (where $v_i$ is a non minitree root node) to the root, say $v_r$, of the minitree $M$ containing $v_i$, retrieve the shortest path length between $v_r$ and $v_s$, and finally add the length of the path between $v_i$ and $v_r$ by using the parent query repeatedly along with retrieving the edge weights of all the edges in between. Note that, we can retrieve these edge weights from the adjacency array while finding the parent only. Thus, this whole process can be completed using $O(\lg n)$ time. If $v_i$ is a root of some minitree, then from $F$ array we can return the answer in $O(1)$ time. To return a shortest path from the queried node $v_i$ to the source $v_s$, we can repeatedly use the parent query from $v_i$ till we reach to $v_s$. Thus this takes time proportional to the length of the path, and this is optimal. Thus, we obtain the following,

\begin{theorem}
Given any sparse (dense, respectively) undirected or directed graph $G$, there exists an $O(m+n \lg n)$ time and $O(n \lg n)$ bits preprocessing algorithm which outputs a data structure of size $O(n)$ ($O(n \lg (m/n))$, respectively) bits, using which the query 1(a) can be reported in $O(\lg n)$ time, and 1(b) can be returned optimally in time proportional to the size of the solution respectively for the \SHORT\ problem.
\end{theorem}

\subsection{Indexing Connectivity in Undirected Graphs}
Now consider the following problem which we call the \UCON\ problem.

\begin{center}
\fbox{\begin{minipage}{11cm}
\noindent{\UCON\ problem}\\
\noindent {\em Input}: An undirected graph $G=(V,E)$ where $|V|=n$, $|E|=m$. Preprocess $G$ and answer the following query:\\
\noindent{\em Query}: Given any pair of vertices, $v_i$ and $v_j$, are they connected in $G$. 
\end{minipage}}
\end{center}

It's easy to see that, by storing the connected component number for every vertex, we can solve this query in $O(1)$ using $O(n \lg n)$ bits of space. We can optimize on space by using our technique. More specifically, if the input graph $G$ is disconnected, we do a DFS of $G$ and store the TC representation for each of the tree in the DFS forest along with an extra piece of information. With every minitree root, we also store the vertex label of the root of the tree (in the DFS forest) where the minitree belongs to. Thus, given any pair of vertices, we just need to reach to the minitree roots containing them, then if the vertex label stored at both these minitrees are same, they belong to the same DFS tree, thus, they are connected. Otherwise, they are disconnected in $G$. Thus, we can solve the required query in $O(\lg n)$ time using Lemma~\ref{reconstruct}, and the final result can be summarized below.

\begin{theorem}
Given any sparse (dense, respectively) undirected graph $G$, there exists an $O(m+n)$ time and $O(n \lg n)$ bits preprocessing algorithm which outputs a data structure of size $O(n)$ ($O(n \lg (m/n))$, respectively) bits, using which the query of the \UCON\ problem can be reported in $O(\lg n)$ time.
\end{theorem}

\subsection{Indexing Strong Connectivity}
A directed graph $G$ is said to be {\it strongly connected} if for every pair of vertices $v_i$ and $v_j$ in $V$, both $v_i$ and $v_j$ are reachable from each other. If $G$ is not strongly connected, it is possible to decompose $G$ into its strongly connected components i.e. a maximal set of vertices $C \subseteq V$ such that for every pair of vertices $v_i$ and $v_j$ in $C$, both $v_i$ and $v_j$ are reachable from each other. Alternatively, if $G$ is directed and $v_i,v_j\in V$, let us write $v_i \equiv_S v_j$ if $G$ contains a path from $v_i$ to $v_j$ and one from $v_j$ to $v_i$. Then it is easy to see that $\equiv_S$ is an equivalence relation on $V$, and each subgraph induced by this equivalence class is called a strongly connected component. Now let us define the following problem which we call the \STRONG\ problem.

\begin{center}
\fbox{\begin{minipage}{11cm}
\noindent{\STRONG\ problem}\\
\noindent {\em Input}: A directed graph $G=(V,E)$ where $|V|=n$ and $|E|=m$, preprocess $G$ and answer the following queries:\\
\noindent{\em Queries}: 
\begin{enumerate}
\item Given $v_i$, return all the vertices that belong to the same strongly connected component component as $v_i$.
\item Given any pair of vertices $v_i$ and $v_j$, check if they belong to the same strongly connected component.
\item Enumerate all the strongly connected components of $G$.
\end{enumerate}
\end{minipage}}
\end{center}

In the preprocessing step, we use a standard algorithm for enumerating strongly connected components as follows.
First we do a DFS on $G$ and store finish time for each vertex using $O(n \lg n)$ bits and
mark roots of the trees in the DFS forest using a bitvector of length $n$.
Then we do a DFS again using reversed edges in decreasing order of finish time, which can be done
using in adjacency array in our graph representation, and store the DFS forest using the same data structure
as other problems.  Each tree in the resulting DFS forest corresponds to a strongly connected component.
We create a virtual root vertex which has roots of DFS trees as children.  The the DFS forest becomes
a virtual DFS tree $T$.  The number of edges increases at most $n$.
We partition the virtual DFS tree $T$ into minitrees using the tree cover algorithm, and for each minitree root,
we store the root node of the DFS tree containing the minitree root using $O(n)$ bits.

Queries are done as follows.  Query 3 is easily solved by finding ones in the bitvector marking
the roots of the DFS trees using {\it select} queries.
For query 1, given a vertex $v_i$, we first climb up the DFS tree
until we hit the root.  Then we do a DFS to enumerate all the vertices in the DFS tree
in time proportional to the tree size.  Because the set of vertices in the tree coincides the strongly connected
component containing $v_i$, we can correctly answer the query.
For query 2, first we climb up the DFS tree from $v_i$ and $v_j$ until we hit a minitree root or
the root of the DFS tree.  If we hit the minitree root, we can obtain the root of the DFS tree.
Therefore we can reach the root of the DFS tree having vertices $v_i$ and $v_j$ in $O(\lg n)$ time.
Then it is easy to check if they belong to the same strongly connected component in constant time
by just comparing the ID's of the roots.

\begin{theorem}
Given any sparse (dense, respectively) directed graph $G$, there exists an $O(m+n)$ time and $O(n \lg n)$ bits preprocessing algorithm which outputs a data structure of size $O(n)$ ($O(n \lg (m/n))$, respectively) bits, using which the query 1 can be answered in time proportional to the size of the solution, 2 can be answered in $O(\lg n)$ time, and finally 3 can be returned optimally in time proportional to the size of the solution.
\end{theorem}

\subsection{Indexing Biconnectivity and $2$-Edge Connecitivity}
Before starting with the next application, let us briefly recollect all the necessary graph theoretic definitions that will be used subsequently. A cut vertex in an undirected graph $G$ is a vertex $v$ that when removed (along with its incident edges) from a graph creates more components than previously in the graph. A (connected) graph with at least three vertices is biconnected (also called $2$-connected in the graph literature sometimes) if and only if it has no cut vertex. A biconnected component is a maximal biconnected subgraph. These components are attached to each other at cut vertices. Similarly in an undirected graph $G$, a bridge (or cut edge) is an edge that when removed (without removing the vertices) from a graph creates more components than previously in the graph. A (connected) graph with at least two vertices is $2$-edge-connected (also called bridgeless sometimes) if and only if it has no bridge. A $2$-edge connected component is a maximal $2$-edge connected subgraph. Alternatively, let $G$ be an undirected graph, and $e_1,e_2 \in E$, the we write $e_1 \equiv_B e_2$ ($e_1 \equiv_E e_2$, respectively) if $e_1=e_2$ or $e_1$ and $e_2$ belong to a common simple cycle (not necessarily simple cycle, respectively) in $G$. Then $\equiv_B$ and $\equiv_E$ are equivalence relations on $E$. Each subgraph induced by an equivalence class of one of these relations is called a biconnected component in the case of $\equiv_B$, and a $2$-edge connected component in the case of $\equiv_E$. In the light of above the definitions, let us define the following problem which we call the \BICON\ problem.

\begin{center}
\fbox{\begin{minipage}{11cm}
\noindent{\BICON\ problem}\\
\noindent {\em Input}: An undirected graph $G=(V,E)$ where $|V|=n$ and $|E|=m$. Preprocess $G$ and answer the following queries:\\
\noindent{\em Queries}: 
\begin{enumerate}
\item Given $v_i$, check if $v_i$ is a cut vertex of $G$.
\item Given an edge $(v_i,v_j)$, return all the edges that belong to the same biconnected component as the edge $(v_i,v_j)$.
\item Given any pair of edges $e_i=(v_a,v_b)$ and $e_j=(v_c,v_d)$, check if both of them belong to the same biconnected component.
\item Enumerate all the cut vertices of $G$.
\end{enumerate}
\end{minipage}}
\end{center}

A similar problem is also studied in~\cite{ChakrabortyRS17,KammerKL16} but in a slightly different setting. More specifically, in~\cite{ChakrabortyRS17,KammerKL16} no preprocessing is allowed. Towards solving the \BICON\ problem, in the preprocessing step, we run Tarjan's~\cite{Tarjan72}   classical biconnectivity algorithm (which takes $O(m+n)$ time and $O(n \lg n)$ bits of space), and in a bit vector $H$ mark all the cut vertices. Then, given any vertex $v_i$, we can check if it is a cut vertex in $O(1)$ time from $H$ for answering query 1. Similary for query 4, using {\it select} query on $H$, we can enumerate all the cut vertices in optimal $O(t)$ time, if there are $t$ cut vertices in $G$. Finally, it is a routine task to peel off the biconnected components by traversing the DFS tree while avoiding the cut vertices (which are explicitly stored in $H$). Thus, we can answer the query 2 in time proportional to the size of the biconnected component where the edge $(v_i,v_j)$ belongs to. Similarly, the query 3 can be answered in time $O(max(|B_i|, |B_j|))$ where $|B_i|$ ($|B_j|$, respectively) is the size of the biconnected component where the edge $e_i$ ($e_j$, respectively) belongs to. Thus, we obtain the following, 

\begin{theorem}
Given any sparse (dense, respectively) undirected graph $G$, there exists an $O(m+n)$ time and $O(n \lg n)$ bits preprocessing algorithm which outputs a data structure of size $O(n)$ ($O(n \lg (m/n))$, respectively) bits, using which the query 1 can be reported in $O(1)$ time, 2 can be answered in time proportional to the size of the solution, 3 can be answered in time proportional to the maximum size of the biconnected components containing the input edges, and finally 4 can be returned optimally in time proportional to the size of the solution respectively for the \BICON\ problem.
\end{theorem}

Similar to the \BICON\ problem, we also define the \EDGE\ problem in the following way.

\begin{center}
\fbox{\begin{minipage}{11cm}
\noindent{\EDGE\ problem}\\
\noindent {\em Input}: An undirected graph $G=(V,E)$ where $|V|=n$ and $|E|=m$, preprocess $G$ and answer the following queries:\\
\noindent{\em Queries}: 
\begin{enumerate}
\item Given an edge $(v_i,v_j)$, 
\begin{enumerate}
\item check if it is a bridge of $G$.
\item return all the edges that belong to the same $2$-edge-connected component as the edge $(v_i,v_j)$.
\end{enumerate}
\item Given any pair of edges $e_i=(v_a,v_b)$ and $e_j=(v_c,v_d)$, check if both of them belong to the same $2$-edge-connected component.
\item Enumerate all the bridges of $G$.
\end{enumerate}
\end{minipage}}
\end{center}

We can return the queries of the \EDGE\ problem almost in an analogous manner. For this, first we note that only the tree edges of the DFS tree could be bridges, thus, we store in an array, say $Y$, all the possible bridges of $G$, and $Y$ is one-to-one correspondence with the unary degree sequence array, child array and the parent array of the DFS tree. Then checking if the edge $(v_i,v_j)$ is a bridge can be done in $O(1)$ time using the {\it select} query. Similarly, enumerating all the bridges can be performed in optimal $O(t)$ time, if there are $t$ bridges in $G$. Also, by running another DFS and explicitly avoiding the bridges, we can peel off the $2$-edge-connected component which contains the edge $(v_i,v_j)$ in time proportional to its size. Finally, we can return the answer of query 2 by first generating the $2$-edge connected component containing $e_i$ and then checking whether $e_j$ belongs there, thus it will take time $O(max(|C_i|, |C_j|))$ where $|C_i|$ ($|C_j|$, respectively) is the size of the $2$-edge connected component where the edge $e_i$ ($e_j$, respectively) belongs to. We can summarize all the results in the following theorem.

\begin{theorem}
Given any sparse (dense, respectively) undirected graph $G$, there exists an $O(m+n)$ time and $O(n \lg n)$ bits preprocessing algorithm which outputs a data structure of size $O(n)$ ($O(n \lg (m/n))$, respectively) bits, using which the query 1(a) can be reported in $O(1)$ time, 1(b) can be answered in time proportional to the size of the solution, 2 can be answered in time proportional to the maximum size of the $2$-edge connected components containing the input edges, and finally 3 can be returned optimally in time proportional to the size of the solution respectively for the \EDGE\ problem.
\end{theorem}

\end{document}